\documentclass[11pt]{article}

\usepackage{color}

\usepackage{amsthm}
\usepackage{amsmath}
\usepackage{amssymb}

\newtheorem{theorem}{Theorem}[section]
\newtheorem{lemma}[theorem]{Lemma}

\newtheorem{corollary}[theorem]{Corollary}
\newtheorem{proposition}[theorem]{Proposition}

\newtheorem{remark}[theorem]{Remark}

\begin{document}

\newcommand{\la}{\triangle}
\newcommand{\bs}{\backslash}
\renewcommand{\d}{\delta}
\newcommand{\D}{\Delta}
\newcommand{\ra}{\rightarrow}
\newcommand{\p}{\partial}
\newcommand{\f}{\frac}
\newcommand{\g}{\gamma}
\newcommand{\G}{\Gamma}
\renewcommand{\l}{\lambda}
\renewcommand{\L}{\Lambda}
\newcommand{\be}{\begin{equation}}
\renewcommand{\ra}{\rightarrow}
\newcommand{\ee}{\end{equation}}
\newcommand{\bea}{\begin{eqnarray}}
\newcommand{\eea}{\end{eqnarray}}
\newcommand{\bna}{\begin{eqnarray*}}
\newcommand{\ena}{\end{eqnarray*}}
\renewcommand{\o}{\omega}
\renewcommand{\O}{\Omega}
\renewcommand{\le}{\left}
\newcommand{\ri}{\right}
\newcommand{\s}{\sigma}
\newcommand{\slb}{\bar{\sum_{\lambda}}}
\newcommand{\ve}{\varepsilon}
\newcommand{\vp}{\varphi}
\openup 1.0\jot
\renewcommand{\theequation}{\arabic{section}.\arabic{equation}}

\title{\bf   Von Neumann entropy and majorization}
\author{Yuan Li$^{a,}$\thanks{E-mail address: liyuan0401@yahoo.com.cn} {\ and}
Paul Busch$^{b,}$\thanks{E-mail address: paul.busch@york.ac.uk}\\
 {\small\em a. College of Mathematics and Information Science, Shaanxi  Normal University,}\\
   {\small\em Xi'an, 710062, People's Republic of China}\\
{\small\em b.  Department of Mathematics, University of York,
York, YO10  5DD,   United Kingdom}
}
\date{}
\maketitle

 \begin{abstract}
 \noindent We consider the properties of the Shannon entropy for two probability 
distributions which stand in the relationship of majorization. Then we give a generalization of a theorem due to 
Uhlmann, extending it to infinite dimensional Hilbert spaces. Finally we show that for any quantum channel
$\Phi$, one has  $S(\Phi(\rho))=S(\rho)$ for all quantum states $\rho$ if and only if there exists an isometric operator 
$V$ such that $\Phi(\rho)=V\rho V^*$.\\

\noindent {\bf  Keywords}:  Von Neumann entropy, majorization, quantum operation.\\
 { AMS Classification}: 47L05, 47L90, 81R10  \\
\end{abstract}

 \section{Introduction}

In this paper we study aspects of the relation of majorization for classical probability distributions and for quantum
mechanical density operators and connections with the action of quantum channels, utilising relevant properties of
the Shannon entropy and von Neumann entropy. We obtain extensions of classic results from finite to infinite dimensional
Hilbert spaces.

Let $\mathcal{B(H)}$ be the von Neumann algebra of all bounded linear operators on a separable Hilbert 
space $\mathcal{H}$ over $\mathbb{C}$ and $S\mathcal{(H)}$ be the set of all {\em density operators} on 
$\mathcal{H}$. That is, $\rho\in S\mathcal{(H)}$ if and only if $\rho\geq0$ and ${\rm tr}(\rho)=1$. The elements of 
$S\mathcal{(H)}$ are taken to represent quantum states in quantum physics, while the selfadjoint elements
of $\mathcal{B(H)}$ represent (bounded) observables. The set $S\mathcal{(H)}$ spans the Banach space  
$\mathcal{T(H)}$ of all trace class operators on $\mathcal{H}$.  

We denote by $\mathcal{B(H,K)}$ the space of all bounded linear operators from $\mathcal{H}$ into another 
Hilbert space $\mathcal{K}$. An operator $V\in\mathcal{B(H,K)}$ is called an isometry if $V^*V=I_{\mathcal{H}}$.
In this case, $VV^*\in {\mathcal{B(K)}}$ is an orthogonal projection. For $x,y\in\mathcal{H}$, $x\otimes y$ denotes 
the (linear, rank-1) operator $z\mapsto \langle z, y\rangle x$ $(z\in\mathcal{H})$. Associated with each $\rho\in
S\mathcal{(H)}$ is the sequence $\lambda(\rho)\equiv \bigl(\lambda_1(\rho),\lambda_2(\rho),\cdots\lambda_n(\rho)\cdots\bigr)$ 
of eigenvalues of $\rho$, arranged in non-increasing order. Thus $\lambda(\rho)\in c_0^*$, where $c_0^*$ is
the positive cone of sequences decreasing monotonically to $0$, as defined in \cite{KW10}. 

Let $l^{\infty}(\mathbb{R})$ and 
$ l_1^1(\mathbb{R}^+)$ denote  the sets of all bounded real sequences and of all summable non-negative real 
sequences which have sum $1$, respectively. For a vector $r \in l^{\infty}(\mathbb{R})$, we introduce 
$r^{\downarrow}=(r_{1}^{\downarrow},\ldots r_{n}^{\downarrow}\cdots)$ as the vector whose elements are the
elements of $r$ re-ordered into non-increasing order.  Adopting the definition of majorization given in 
\cite{Ante07},\cite{KW10}, for $r,s\in {c_0^*}$, we say that $r$ is majorized by $s$, written as $r\prec s$, 
if 
\begin{displaymath}
\sum_{i=1}^{k}r_i^{\downarrow}\leq\sum_{i=1}^{k}s_i^{\downarrow},\hbox { for  } k=1,2,\cdots
  \qquad \hbox {and }
 \sum_{i=1}^{\infty}r_i^{\downarrow}=\sum_{i=1}^{\infty}s_i^{\downarrow}.
 \end{displaymath}
 
Recently, 
the infinite dimensional Schur-Horn theorem and infinite majorization have received much attention. 
In \cite{Neumann99}, A.\ Neumann has given properties of infinite majorization in $l^{\infty}(\mathbb{R})$; 
V.\ Kaftal and G.\ Weiss \cite{KW10} have obtained interesting results for infinite majorization in $c_0^*$.  
Arveson and Kadison \cite{AK06} presented some other characterizations by different methods.

Some results relevant for our purposes are the following: if $r,s\in l_1^1(\mathbb{R}^+)$, then
$$
r\prec s\Longleftrightarrow r=Qs,\hbox{ with } Q_{ij}=|U_{ij}|^2 \hbox{ for some unitary U } \ \hbox{\cite[Theorem 1]{GM64}},
$$
and 
$$
r\prec s\Longleftrightarrow r=Qs,\hbox{ for some orthostochastic matrix Q}\ \hbox{\cite[Corollary 6.1]{KW10}}.
$$
Motivated by the above studies, we first consider the properties of the Shannon entropy for two elements in 
$l_1^1( \mathbb{R}^+)$ that satisfy majorization. Then we extend and study those properties for two density operators 
$\rho,\sigma\in S\mathcal{(H)}$. Following the definition given for finite dimensional spaces (see \cite{AU82}), we denote
$\rho\prec \sigma$ for two operators $\rho,\sigma\in S\mathcal{(H)}$ if $\lambda(\rho)\prec\lambda(\sigma)$.

Let ${\mathcal{M}}_n({\mathcal{B(H)}})$ be the von Neumann algebra of $n\times n$ matrices
whose entries are in ${\mathcal{B(H)}}$,  and let $\Phi:$ ${\mathcal{B(H)}}\to {\mathcal{B(H)}}$  
be a linear map. Then $\Phi$ induces a map 
$\Phi_n:$${\mathcal{M}}_n({\mathcal{B(H)}})\to {\mathcal{M}}_n({\mathcal{B(H)}})$ by the formula
$$
\Phi_n((a_{i,j}))=(\Phi(a_{i,j})) \hbox {   for  }  (a_{i,j})\in{\mathcal{M}}_n({\mathcal{B(H)}}).
$$
If every $\Phi_n$ is a positive map, then $\Phi$ is called completely positive.
$\Phi$ is said to be normal if $\Phi$ is continuous with respect to the ultraweak ($\sigma-$weak) topology.
Normal completely positive contractive maps on ${\mathcal{B(H)}}$ were characterized in a theorem of 
Kraus \cite[Theorem 3.3]{Kraus71}, which says that $\Phi$ is a normal completely positive map if and only if there 
exists a sequence $\{A_i\}_{i=1}^\infty $ in $\mathcal{B(H)}$ such that for all $X\in\mathcal{B(H)}$, 
$$
\Phi(X)=\sum_{i=1}^{\infty}A_iXA_i^*  \hbox {   } \hbox {   } \hbox {   } \hbox {   } \hbox {  with  } \hbox {   } 
\sum_{i=1}^{\infty}A_iA_i^*\leq I.
$$
where the limits are defined in the strong operator topology. The sequence $\{A_i\}_{i=1}^\infty $, which is not unique, 
is also called a family of Kraus operators for $\Phi$. In this case, $\Phi$ has a dual  map $\Phi^\dag$ defined by
$$
\Phi^\dag(X)=\sum_{i=1}^{\infty}A_i^*XA_i \hbox {\quad for } X\in\mathcal{T(H)},
$$
where the sum converges in the trace norm topology.
It is easy to see that one has 
$$
|{\rm tr}[\Phi^\dag(X)Y]|=|{\rm tr}[X\Phi(Y)]|\leq \|\Phi\|\|Y\|\,{\rm tr}(|X|),\quad X\in\mathcal{T(H)},\ 
Y\in{\mathcal{B(H)}},
$$ 
so $\Phi^\dag(X)\in\mathcal{T(H)}$ and $\Phi^\dag$ is well defined on $\mathcal{T(H)}$. In general,
$\Phi^\dag$ cannot be extended from $\mathcal{T(H)}$ into $\mathcal{B(H)}$.
However, if  $\sum_{i=1}^{\infty}A_i^*A_i\leq I$, then $\Phi^\dag$ is well defined as a normal map on $\mathcal{B(H)}$.
A normal completely positive map $\Phi$ which is trace preserving ($\Phi^\dag (I)=I$, corresponding to 
${\rm tr}\bigl(\Phi(X)\bigr)={\rm tr}(X)$
for $X\in\mathcal{T(H)}$)
is called a {\em quantum channel}. If a normal completely positive map satisfies $\Phi(I)\leq I$, then $\Phi$ is
called a quantum operation \cite{Kraus71,Li11}. A quantum operation is unital if $ \Phi(I)=I$, which is equivalent to  
$ \sum_{j}A_{j}A_{j}^*=I$. A quantum operation is bi-stochastic if it is both trace-preserving and unital. 
In particular, $\Phi$ is said to be a mixed unitary operation if $\Phi(X)=\sum_{i=1}^{n}t_iU_iXU_i^*$, 
where $n<\infty$, the $U_i$ are all unitary operators and $t_i>0$, $\sum_{i=1}^{n}t_i=1$.

The von Neumann entropy of a quantum state $\rho$ is defined by the formula 
$$
S(\rho)\equiv-{\rm tr}(\rho\log(\rho)).
$$ 
Here we follow the common practice of taking the logarithm to base two. In classical information theory, 
the Shannon entropy is defined
by $H(p)= -\sum_{i}p_i \log(p_i)$, where $p=(p_1,p_2\cdots p_n\cdots)$ is a probability distribution. 
If $\lambda_{i}$ are the eigenvalues of $\rho$, then the von Neumann entropy can be re-expressed as 
$$
S(\rho)=-\sum_{i=1}^\infty\lambda_{i}\log(\lambda_{i})=H(\lambda(\rho)),
$$ 
where we use $0\log0=0$.

There are extensive recent studies of quantum operations which  preserve the von Neumann entropy and the 
relative entropy of quantum states  \cite{HMPB11,LiWang12,Petz08,ZW11}

Hardy, Littlewood and P$\acute{\rm o}$lya \cite{HLP73} showed for $\xi,\eta\in{ \mathbb{R}}^n$,
$$
\xi\prec\eta\Longleftrightarrow\xi=Q\eta \hbox{ for some doubly stochastic matrix Q}.
$$ 
In the quantum context, Uhlmann proved the following for any pair of density operators 
$\rho,\sigma\in S\mathcal{(H)}$ acting in a finite dimensional Hilbert space $\mathcal{H}$:
$$
\rho\prec \sigma\Longleftrightarrow \rho=\Phi(\sigma), \hbox{ for some mixed unitary quantum operation } \Phi.
$$ 
Uhlmann's theorem can be used to study the role of majorization in quantum mechanics. 

Here we first consider the properties of the Shannon entropies of two probability distributions which obey 
majorization. Then we give a generalization of  Uhlmann's theorem for infinite dimensional Hilbert 
spaces. Finally, we give a characterization of quantum channels that leave the von Neumann entropy  invariant.

\section{Shannon entropy of infinite probability distributions}

The following lemma is a direct corollary of \cite[Theorem 8.0.1]{Nielsen02} 
as the function $f(x)\equiv -x\log(x)$ is a concave function.
\begin{lemma}\label{lem:maj-finite}
Let $a,b\in{\mathbb{R}}^n$ with $a_i,b_i\geq 0$, and $a\prec b$.  Then $H(a)\geq H(b)$.
\end{lemma}
This result extends to the infinite case as follows.
\begin{theorem}\label{thm:maj-infinite}
Let $a,b\in l_1^1( \mathbb{R}^+)$ and $a\prec b$.  Then $H(a)\geq H(b)$.
\end{theorem}

\begin{proof}
Suppose that $c=(c_1,c_2,\cdots c_n,\cdots)$, where $c_i\geq 0$ and
$\sum_{i=1}^\infty c_i\leq 1$. It is clear that for arbitrary $n$,
we have 
\begin{equation*} 
(c_1,c_2,\cdots c_n,0,0,0,\cdots)\prec\biggl(\sum_{i=1}^nc_i,0,0,0,\cdots\biggr),
\end{equation*} 
so  Lemma \ref{lem:maj-finite} implies $-\sum_{i=1}^nc_i\log c_i\geq-(\sum_{i=1}^nc_i)\log(\sum_{j=1}^nc_j)$.
Letting $n\to \infty$, then by the continuity of the function $f(x)=-x\log x$ we get
\begin{equation}\label{eq:c-log}
-\sum_{i=1}^\infty c_i\log c_i\geq-\sum_{i=1}^\infty c_i\log\bigl(\sum_{j=1}^\infty c_j\bigr).
\end{equation}
We assume first that $b$ has only finitely many nonzero elements. Without loss of generality,
we suppose 
\begin{equation*} 
(a_1,a_2,\cdots a_n,\cdots)=a\prec b=(b_1,b_2,\cdots,b_m,0,0\cdots ),
\end{equation*}
where $b_i>0$, for $1\leq i\leq m$. Thus there exists $N$ such that
$\sum_{i=N+1}^\infty a_i<\min\{b_1,b_2,\cdots b_m\}$, so
\begin{equation*} 
\biggl(a_1,a_2,\cdots a_N,\sum_{i=N+1}^\infty a_i,0,0,\cdots\biggr)\prec (b_1,b_2,\cdots,b_m,0,0\cdots ),
\end{equation*}
which implies
\begin{equation*}
-\sum_{i=1}^N a_i\log a_i-\sum_{i=N+1}^\infty a_i \log\biggl(\sum_{j=N+1}^\infty a_j\biggr)\geq-\sum_{i=1}^m b_i\log(b_i).
\end{equation*}
By inequality \eqref{eq:c-log}, we know $H(a)\geq H(b)$.

For the general case of $b$, we note that for any positive integers $s$, the relation $a\prec b$ entails
\begin{equation*} 
(a_1,a_2,\cdots a_n,\cdots)\prec \biggl(b_1,b_2,\cdots,b_s,\sum_{i=s+1}^\infty b_i,0,0\cdots \biggr),
\end{equation*}
so the proof above yields  
\begin{equation*}
H(a)\geq-\sum_{i=1}^s b_i\log b_i-\sum_{i=s+1}^\infty b_i \log\biggr(\sum_{j=s+1}^\infty b_j\biggr). 
\end{equation*}
Letting $s\to \infty$, we obtain $H(a)\geq H(b)$, as desired.
\end{proof}

\begin{lemma}\label{lem:maj-fin2} 
Let $a,b\in{\mathbb{R}}^n$ with $a_i,b_i\geq 0$. If $a\prec b$ and $H(a)=H(b)$,
then $a^\downarrow=b^\downarrow$.
\end{lemma}

\begin{proof}
Without loss of generality, we assume that
\begin{equation*} 
a^\downarrow=(a_1,a_2,\cdots a_m,0\cdots 0)\prec b^\downarrow=(b_1,b_2,\cdots,b_m,0\cdots 0),
\end{equation*}
where $m\leq n$ and $a_i+b_i\neq 0$, for $1\leq i\leq m$.
For $0<t<1$, denote 
$$
c_t=\bigl(ta_1+(1-t)b_1,ta_2+(1-t)b_2,\cdots ta_m+(1-t)b_m\bigr).
$$ 
It is easy to verify that $a\prec c_t\prec b$, so Lemma \ref{lem:maj-finite} implies 
\begin{equation}\label{eq:ha}
-\sum_{i=1}^m\bigl(ta_i+(1-t)b_i)\log(ta_i+(1-t)b_i\bigr)=H(a). 
\end{equation}
Taking the second derivative with respect to $t$ on both sides of  equation \eqref{eq:ha}, we get
$$
\sum_{i=1}^m\frac{(a_i-b_i)^2}{ta_i+(1-t)b_i}=0,
$$
which yields $a_i=b_i$ for $1\leq i\leq m$, so $a^\downarrow=b^\downarrow$.
\end{proof}

The following proposition shows that the Shannon entropy of an infinite probability distribution is strictly 
monotone with respect to the relation of majorization.

\begin{proposition}\label{prop:ha=hb}
Let $a,b\in l_1^1( \mathbb{R}^+)$ such that all $a_ib_i>0$.
If $a\prec b$ and $H(a)=H(b)<\infty$, then $a^\downarrow=b^\downarrow$.
\end{proposition}
\begin{proof}
Let 
$$
a^\downarrow=(a_1,a_2,\cdots a_n, \cdots)\prec b^\downarrow=(b_1,b_2,\cdots,b_n,\cdots).
$$
Then $a_1\leq b_1$, so there exists $k$ such that $a_1\in (b_{k+1},b_k]$, which implies
$a_1=tb_{k+1}+(1-t)b_k$, for some $0\leq t\leq 1$. Denote
$$
c=\bigl(b_1,b_2,\cdots tb_{k+1}+(1-t)b_k,(1-t)b_{k+1}+tb_k, b_{k+2},\cdots\bigr),
$$
it is clear that $a\prec c\prec b$. Thus by Theorem \ref{thm:maj-infinite} and the assumption, we have
$H(a)=H(b)=H(c)$,
so Lemma \ref{lem:maj-fin2} implies $a_1=tb_{k+1}+(1-t)b_k=b_k$.
Let 
$$
\tilde{a}=(a_2,a_3,\cdots a_n, \cdots)\hbox{ and } \tilde{b}=(b_1,b_2,\cdots b_{k-1},b_{k+1},\cdots).
$$
Thus $\tilde{a}\prec \tilde{b}$ and $H(\tilde{a})=H(\tilde{b})<\infty$,
so by a similar proof, we conclude that there exists $k_1>k$ such that
$a_2=b_{k_1}$.
By mathematical induction, we get $a_{n+1}=b_{k_n}$, for some subsequence of $b$.
However we have
$\sum_{i=1}^\infty a_i=\sum_{i=1}^\infty b_i=1$, and then the fact that  all $a_ib_i>0$ implies 
$a_{n+1}=b_{k_n}=b_{n+1}$, that is $a^\downarrow=b^\downarrow$.
\end{proof}

\begin{remark}
The conditions that all $a_ib_i>0$ and $H(a)<\infty$ are essential in
Proposition 2.4. Indeed,
it is obvious that both conditions $a\prec b$ and $H(a)=H(b)<\infty$ are not changed
if we add some zeros for $a$. Furthermore, if $H(a)=H(b)=\infty$,
we may replace $a$ by $a'=(\frac{a_1}{2},\frac{a_1}{2},a_2,\cdots)$.
Then $a'\prec b$ and $H(a')=H(b)=\infty$.
However, it is a contradiction that $a^\downarrow=b^\downarrow=a'^\downarrow$.
\end{remark}

 \section{Von Neumann entropy of quantum states}

Let us recall that two quantum states $\rho_1$ and $\rho_2$ are ${\mathcal{L}}^1$-{\em equivalent} 
if and only if there is a sequence of unitary operators $\{U_i\}_{i=1}^\infty$
such that $\lim_{n\to\infty}\|\rho_1-U_n\rho_2U_n^*\|_1=0$.

 \begin{proposition}
Let $\rho_1,\rho_2\in S({\mathcal{H}})$. If $\rho_1\prec \rho_2$ and $S(\rho_1)=S(\rho_2)<\infty$,
then $\rho_1$ and $\rho_2$ are ${\mathcal{L}}^1$-equivalent.
\end{proposition}
\begin{proof}
By Proposition \ref{prop:ha=hb}, all the non-zero spectral points (including multiplicity) of
$\rho_1$ and $\rho_2$ are the same, which is equivalent to the fact
 that $\rho_1$ and $\rho_2$ are ${\mathcal{L}}^1$-equivalent, by \cite[Proposition 3.1]{AK06}.
 \end{proof}

The following two results give extensions of Uhlmann's theorem in an infinite dimensional Hilbert space. We need some eigenvalue estimates of Weyl given in \cite{Weyl49}. Let $A$ be a compact operator with eigenvalues
$\lambda_1\geq\lambda_2\geq\lambda_3\geq\cdots$ and ${ \mathcal{P}}_n$ be the set of all $n$-dimensional projections. Then
$\max_{P\in{ \mathcal{P}}_n}{\rm tr}(AP)=\sum_{i=1}^n\lambda_i$.

\begin{proposition}
Let $\rho_1,\rho_2\in S({\mathcal{H}})$. If $\rho_1$ is finite rank, then $\rho_1\prec \rho_2$
if and only if there exists a mixed unitary operation $\Phi$ such that $\Phi(\rho_2)=\rho_1$.
\end{proposition}
\begin{proof}
Necessity. If $\rho_1$ is finite rank, then denote 
$\lambda(\rho_1)=\bigl(\lambda_1(\rho_1),\lambda_2(\rho_1)\cdots \lambda_m(\rho_1),0,0\cdots\bigr)$.
As $\rho_1\prec \rho_2$, we have that $\rho_2$ is also finite rank and
$\lambda(\rho_2)=\bigl(\lambda_1(\rho_2),\lambda_2(\rho_2)\cdots \lambda_n(\rho_2),0,0\cdots \bigr)$ where $n\leq m$.
Using the spectral decomposition of the states $\rho_1$ and $\rho_2$, we conclude that there exist
two orthonormal bases $\{x_i\}_{i=1}^{\infty}$ and $\{y_i\}_{i=1}^{\infty}$ of ${\mathcal{H}}$ such that
$$
\rho_1=\sum_{i=1}^m\lambda_i(\rho_1)x_i\otimes x_i \hbox{ and }  \rho_2=\sum_{i=1}^n\lambda_i(\rho_2)y_i\otimes y_i.
$$
Let ${\mathcal{H}}_1$ be the subspace spanned by $\{x_i\}_{i=1}^{m}$ and $U$ be the unitary operator
defined by $Uy_i=x_i$, for $i=1,2\cdots$. Then $\rho_1|_{{\mathcal{H}}_1}$ and
$U\rho_2U^*|_{{\mathcal{H}}_1}$ can be represented by $m\times m$ matrices which have spectra
$\bigl(\lambda_1(\rho_1),\lambda_2(\rho_1)\cdots \lambda_m(\rho_1)\bigr)$ and 
$\bigl(\lambda_1(\rho_2),\lambda_2(\rho_2)\cdots \lambda_n(\rho_2),\cdots 0\bigr)\in{\mathbb{R}}^m$, respectively. 
Thus Uhlmann's theorem (\cite{AU82} or \cite[Theorem 4.1.1]{Nielsen02}) implies that there exists a mixed unitary operation
$\widetilde{\Phi}$ on ${\mathcal{B(H}}_1)$, that is
$\rho_1|_{{\mathcal{H}}_1}=\sum_{i=1}^mt_iV_i(U\rho_2U^*)|_{{\mathcal{H}}_1}V_i^*$, where 
$V_i\in{\mathcal{ B(H}}_1)$ are unitary operators, $\sum_{i=1}^mt_i=1$ and  all $t_i>0$. 
As $(U\rho_2U^*)|_{{\mathcal{H}}_1^\perp}=0$, we denote $U_i={\rm diag}(V_i,I_{{\mathcal{H}}_1^\perp})$,
so $U_i\in {\mathcal{B(H}})$ are unitary operators, and $\rho_1=\sum_{i=1}^mt_iU_i(U\rho_2U^*)U_i^*$.

Sufficiency. By Weyl's estimates \cite{Weyl49},
we get for $n=1,2\cdots$,   
\begin{align*}
\sum_{i=1}^n \lambda_i(\rho_1)&=\max_{P\in{ \mathcal{P}}_n}\{{\rm tr}(\rho_1P)\} 
=\max_{P\in{\mathcal{P}}_n}\biggl\{{\rm tr}\biggl(\sum_{i=1}^m t_iU_i\rho_2 U_i^*P\biggr)\biggr\} \\
&=\max_{P\in{\mathcal{P}}_n}\biggl\{\sum_{i=1}^m t_i{\rm tr}(\rho_2 U_i^*PU_i)\biggr\}  
\leq\sum_{i=1}^m t_i\max_{P\in{\mathcal{P}}_n}\{{\rm tr}(\rho_2 U_i^*PU_i)\} \\
&= \max_{P\in{\mathcal{P}}_n}\{{\rm tr}(\rho_2P)\} 
=\sum_{i=1}^n \lambda_i(\rho_2).
\end{align*}
\end{proof}

\begin{theorem}\label{thm:maj-equiv}
Let $\rho_1,\rho_2\in S({\mathcal{H}})$.  Then the following three conditions are equivalent:
\begin{itemize}
\item[{\rm (a)}] $\rho_1\prec \rho_2$.
\item[{\rm (b)}] There exists a sequence of mixed unitary operation $\Psi_n$ and a bi-stochastic quantum operation $\Psi$ on $S({\mathcal{H}})$ such that $\lim_{n\to\infty} \|\Psi_n(\rho)-\Psi(\rho)\|_1=0$ for all $\rho\in S({\mathcal{H}})$, and $\Psi(\rho_2)=\rho_1$.
\item[{\rm (c)}] There exists a bi-stochastic quantum operation $\Psi$ such that $\Psi(\rho_2)=\rho_1$.
\end{itemize}
\end{theorem}
\begin{proof}
$(a)\Rightarrow (b)$. If $\rho_1\prec \rho_2$, then
$\lambda(\rho_1)^\downarrow\prec\lambda(\rho_2)^\downarrow$, it follows from \cite[Theorem 1]{GM64} that
$\lambda(\rho_1)^\downarrow=Q\lambda(\rho_2)^\downarrow$, where $Q$ is an infinite matrix satisfying 
$Q_{ij}=|u_{ij}|^2$ for the elements of a unitary matrix $U$.
By virtue of the spectral decomposition of $\rho_2$, there exists an orthonormal basis $\{y_i\}_{i=1}^{\infty}$ of ${\mathcal{H}}$ 
such that $\rho_2=\sum_{i=1}^\infty\lambda_i(\rho_2)y_i\otimes y_i$. 
Defining a unitary operator $\widetilde{U}\in B({\mathcal{H}})$ via $\langle\widetilde{U}y_i, y_j\rangle=u_{ij}$ for $i,j=1,2\cdots$, 
we set $e_i=\widetilde{U}y_i$ for $i=1,2\cdots$. It is easy to see that $\{e_i\}_{i=1}^{\infty}$ is
an orthonormal basis of ${\mathcal{H}}$ such that
$\lambda_i(\rho_1)=\langle\rho_2e_i,e_i\rangle$ for $i=1,2\cdots$. With respect to the basis $\{e_i\}_{i=1}^{\infty})$, then 
$\rho_2$ has an infinite matrix form given as follows:
\begin{equation*}
\rho_2 =\left(\begin{array}{ccccc}  \lambda_1(\rho_1) &  \lambda_{12}&\cdots& \lambda_{1n}&\cdots\\
\lambda_{21}&\lambda_2(\rho_1)&\cdots& \lambda_{2n}&\cdots\\
\vdots&\vdots&\ddots&\vdots& \vdots\\
\lambda_{n1}&\cdots &\cdots&\lambda_n(\rho_1)  &\cdots\\
\vdots&\vdots& \vdots&\vdots& \ddots\end{array}\right),
\end{equation*} 
where
$\lambda_{ij}=\overline{\lambda_{ji}}$. Denoting the sequence of rank-one projections $E_i\equiv e_i\otimes e_i$,  then define
$\Phi(\rho)=\sum_{i=1}^\infty E_i\rho E_i$, for $\rho\in S({\mathcal{H}})$.
Also for $n=1,2\cdots$, let $\Phi_n$ be the mixed unitary operation defined as
$\Phi_n(\rho)=\frac{1}{n}\sum_{i=1}^nU^i\rho (U^i)^*$, where $U={\rm diag}(\omega,\omega^2\cdots\omega^n,1,1\cdots 1\cdots)$
and $\omega=e^{\frac{2\pi i }{n}}$.
By a direct calculation we obtain 
$$
\Phi_n(e_i\otimes e_j)=\frac{1}{n}\sum_{i=1}^nU^i(e_i\otimes e_j)(U^i)^*
= 
\begin{cases}
e_i\otimes e_i & \text{ if }  i=j \\ 
e_i\otimes e_j & \text{ if } i,j>n-1\\
0  &\text{ if } i\neq j \text{ and } i\leq n-1 \text{ or } j\leq n-1.
\end{cases}
$$
Thus $\Phi_n(\rho_2)$ becomes the following block diagonal matrix with respect to the subspace decomposition
$\mathcal{H}=\bigvee_{i=1}^{n-1} \{e_i\}\bigoplus  \bigvee_{i=n}^\infty \{e_i\}$:
\begin{equation*}
\Phi_n(\rho_2)
=\left(\begin{array}{cc}\rho_{21} & 0\\ 0& \rho_{22}\end{array}\right),
\end{equation*}
where 
\begin{equation*}
\rho_{21}
={\rm diag}\bigl( \lambda_1(\rho_1), \lambda_2(\rho_1)\cdots \lambda_n(\rho_1)\bigr),
\quad \rho_{22}=(I-P_{n})\rho_{2}(I-P_{n})|_{(\bigvee_{i=n+1}^\infty \{e_i\} )},
\end{equation*} 
and $P_{n}$ denotes the orthogonal projection on the subspace 
$\bigvee_{i=1}^{n-1} \{e_i\}$ spanned by $\{e_i\}_{i=1}^{n-1}$.
Then 
$$
\|\Phi_n(\rho_2)-\Phi(\rho_2)\|_1
=2\,{\rm tr} \bigl([\Phi_n(\rho_2)-\Phi(\rho_2)]^{+}\bigr)\leq 2\,{\rm tr}\bigl(\rho_{22}\bigr )\to 0 \text{ as }n\to\infty,
$$ 
since
${\rm tr}\bigl(\Phi_n(\rho_2)\bigr)={\rm tr}\bigl(\Phi(\rho_2)\bigr)$ implies
\begin{align*}
\|\Phi_n(\rho_2)-\Phi(\rho_2)\|_1&={\rm tr}\bigl(|\Phi_n(\rho_2)-\Phi(\rho_2)|\bigr) \\
&=
{\rm tr}\bigl([\Phi_n(\rho_2)-\Phi(\rho_2)]^{+}\bigr)+{\rm tr}\bigl([\Phi_n(\rho_2)-\Phi(\rho_2)]^{-}\bigr) \\
&=
2\, {\rm tr}\bigl([\Phi_n(\rho_2)-\Phi(\rho_2)]^{+}\bigr);
\end{align*}
and $\Phi_n(\rho_2)-\Phi(\rho_2)\leq(I-P_{n})\rho_{2}(I-P_{n})$ yields
$[\Phi_n(\rho_2)-\Phi(\rho_2)]^+\leq P_{+}(I-P_{n})\rho_{2}(I-P_{n})P_{+}$,  so
\begin{align*}
{\rm tr}\bigl([\Phi_n(\rho_2)-\Phi(\rho_2)]^+\bigr)&\leq {\rm tr}\bigl(P_{+}(I-P_{n})\rho_{2}(I-P_{n})P_{+}\bigr) \\
& \leq {\rm tr}\bigl((I-P_{n})\rho_{2}(I-P_{n})\bigr) \\
&={\rm tr}\bigl(\rho_{22}\bigr),
\end{align*}
Here we use the notation $A^+$, $A^-$ for the positive and negative parts of the self-adjoint operator $A$, and
$P_{+}$ is the orthogonal projection on the range of $[\Phi_n(\rho_2)-\Phi(\rho_2)]^+$.
From the spectral decomposition of $\rho_1$ we have
an orthonormal basis $\{f_i\}_{i=1}^{\infty}$ of ${\mathcal{H}}$ such that
$\rho_1=\sum_{i=1}^\infty\lambda_i(\rho_1)f_i \otimes f_i$.  Define a unitary operator $V$ by $Vf_i=e_i$, for $i=1,2\cdots$, so
$\rho_1=V^*\Phi(\rho_2)V$.  Denote
$\Psi_n(\rho)=V^*\Phi_n(\rho)V$ and $\Psi(\rho)=V^*\Phi(\rho)V$ for all $\rho\in S({\mathcal{H}})$. Then
by a similar proof to the above, we obtain
$\lim_{n\to\infty} \|\Psi_n(\rho)-\Psi(\rho)\|_1=\lim_{n\to\infty} \|\Phi_n(\rho)-\Phi(\rho)\|_1=0$ for all $\rho\in S({\mathcal{H}})$.

$(b)\Rightarrow (c)$ is clear.

$(c)\Rightarrow (a)$. Let $\rho_1=\Psi(\rho_2)=\sum_{i=1}^\infty A_i\rho_2 A_i^*$. 
Applying Weyl's eigenvalue theorem, we get for $n=1,2\cdots$,  
\begin{equation}\label{eq:lambda-sum}
\sum_{i=1}^n \lambda_i(\rho_1)=\max_{P\in{ \mathcal{P}}_n}\biggl\{{\rm tr}\biggl(\sum_{i=1}^\infty A_i\rho_2 A_i^*P\biggr)\biggr\}=
\max_{P\in{ \mathcal{P}}_n}\biggl\{{\rm tr}\biggl(\rho_2\sum_{i=1}^\infty A_i^*PA_i\biggr)\biggr\}.
\end{equation}
As $\sum_{i=1}^\infty A_i^*A_i=I$ and $\sum_{i=1}^\infty A_iA_i^*=I$, then $\sum_{i=1}^\infty A_i^*PA_i\leq I$ and
$$
{\rm tr}\biggl(\sum_{i=1}^\infty A_i^*PA_i\biggr)={\rm tr}\biggl(\sum_{i=1}^\infty A_iA_i^*P\biggr)=n.
$$ 
For convenience,
we denote $B=\sum_{i=1}^\infty A_i^*PA_i$ and $u_i=\langle By_i, y_i\rangle$,
where $\{y_i\}_{i=1}^{\infty}$ is an orthonormal basis of ${\mathcal{H}}$ satisfying 
$\rho_2=\sum_{i=1}^\infty\lambda_i(\rho_2)y_i \otimes y_i$. Thus
$0\leq u_i\leq 1$ and $\sum_{i=1}^\infty u_i=n$, so
\begin{equation}\label{eq:trrho2b}
\begin{split}
{\rm tr}(\rho_2B)=\sum_{i=1}^\infty\lambda_i(\rho_2)u_i
&=
\sum_{i=1}^n\lambda_i(\rho_2)u_i+\sum_{i=n+1}^\infty\lambda_i(\rho_2)u_i \\
&\leq
\sum_{i=1}^n\lambda_i(\rho_2)u_i+\lambda_{n+1}(\rho_2)\sum_{i=n+1}^\infty u_i \\
&\leq
\sum_{i=1}^n\lambda_i(\rho_2)u_i+\lambda_{n+1}(\rho_2)\sum_{i=1}^n(1-u_i) \\
&\leq
\sum_{i=1}^n\lambda_i(\rho_2).
\end{split}
\end{equation} 
By equation \eqref{eq:lambda-sum} and inequality \eqref{eq:trrho2b}, we conclude that for all $n=1,2\cdots$,
one has the inequalities $\sum_{i=1}^n \lambda_i(\rho_1)\leq \sum_{i=1}^n \lambda_i(\rho_2)$, as desired.
\end{proof}

\begin{remark}
In a finite dimensional Hilbert space, $\rho_1\prec \rho_2$ is equivalent to $\rho_1=\Phi(\rho_2)$, 
for some mixed unitary operation $\Phi$.
However, for an infinite dimensional Hilbert space, $\rho_1\prec \rho_2$ does not imply 
$\rho_1=\sum_{i=1}^\infty t_iU_i\rho_2U_i^*$, where
$\sum_{i=1}^\infty t_i=1$, $t_i\geq 0$ and $U_i$ are unitary operators for all $i$. Indeed, the condition 
$\rho_1=\sum_{i=1}^\infty t_iU_i\rho_2U_i^*$
yields $\dim (\ker(\rho_1))\leq \dim (\ker(\rho_2))$. But we can supplement many zeros for $\lambda(\rho_1)$.
\end{remark}

The following proposition was obtained in \cite{LiWang12,ZW11} for the finite case,
in which the condition of injectivity of $\rho$ may be dropped.

\begin{proposition}
Let $\rho\in S({\mathcal{H}})$ and $\Phi$ be a bi-stochastic quantum operation.
If $\rho$ is injective and $S(\rho)=S(\Phi(\rho))<\infty$, then $\Phi(\rho)=U\rho U^*$ for a unitary operator $U$.
\end{proposition}
\begin{proof} 
Suppose $\Phi(\rho)=\sum_{i=1}^\infty A_i\rho A_i^*$. Then by Theorem \ref{thm:maj-equiv}, we have $\Phi(\rho)\prec \rho$.
We claim that if $\rho$ is injective,  then so is $\Phi(\rho)$. Thus, suppose
 that $\Phi(\rho)$ were not injective, then there is a vector $x\neq 0$ that satisfies $\Phi(\rho)x=0$, so
$\bigl\langle\sum_{i=1}^\infty A_i\rho A_i^*x, x\bigr\rangle=0$, which yields $A_i^*x=0$, for all $i$.
Thus $x=\sum_{i=1}^\infty A_i A_i^*x=0$, which is a contradiction.
It follows from Proposition \ref{prop:ha=hb} that $\lambda(\rho)^\downarrow=\lambda(\Phi(\rho))^\downarrow$, so
by the spectral decomposition theorem, we get $\Phi(\rho)=U\rho U^*$ for a unitary operator $U$.
\end{proof}

\begin{remark}
In an infinite dimensional Hilbert space, the condition that $\rho$ is injective may not be dropped.
\end{remark}

In the following, we shall characterize the structure of a quantum channel that does not change the
von Neumann entropy of any quantum state. In \cite{MS10}, Moln\'{a}r and Szokol gave
the structure of the map which preserves the relative entropy in a finite-dimensional Hilbert space.
However, we were unable to find a publication that studied the structure of a quantum channel which 
preserves the von Neumann entropy. 
After a preprint version of the present paper was published as arXiv:1304.7442(v1), we received a 
sketch of a shorter proof of Theorem \ref{thm:isom} from M.B.~Ruskai based on the Stinespring representation theorem 
and some techniques from quantum information. 
Here we mainly use methods of operator theory and Kraus's theorem. The following two lemmas are needed.
The set of all compact operators on $\mathcal{H}$ is denoted $\mathcal{K(H)}$.

\begin{lemma}\label{lem:S-props}
(See {\rm \cite{NC00}}.) Let $\rho\in S(\mathcal{H})$.
\begin{itemize}
\item[{\rm (i)}] 
$S(\rho)\geq 0$, and $S(\rho)= 0$ if and only if $\rho$ is a pure state.

\item[{\rm (ii)}] If $\dim(\mathcal{H})=n$, then $S(\rho) \leq \log n$, and $S(\rho)=\log n$ if and only if $\rho=\frac{1}{n}I$.
\end{itemize}
\end{lemma}

\begin{lemma}\label{lem:commut}
(See {\rm \cite{Li11}}.) Let $\Phi$ be a quantum operation with 
a set of Kraus operators $K(\Phi)\equiv\{A_i\in {\mathcal{B(H)}},  i=1,2\cdots \}$ and such that $\Phi^\dag(I)\leq I$. Then
$\{B\in {\mathcal{K(H)}}:\Phi(B)=B\}=\{B\in {\mathcal{K(H)}}:\Phi^\dag(B)=B\}\subseteq {\mathcal{A}}'$, where $\mathcal{A}'$
is the commutant of $\mathcal{A}=\{A_i,A_i^*: A_i\in {K(\Phi)},  i=1,2\cdots \}$.
\end{lemma}

\begin{theorem}\label{thm:isom}
Let $\Phi$ be a quantum channel on ${\mathcal{B(H}})$.
Then $S(\Phi(\rho))=S(\rho)$ for all quantum states $\rho\in S({\mathcal{H}})$
if and only if there exists an isometry $V\in{\mathcal{ B(H}})$ such that
$\Phi(X)=VXV^*$ for all $X\in{\mathcal{ B(H}})$.
\end{theorem}

\begin{proof}
Sufficiency is clear.

Necessity. For any unit vector $x$, we have $S(\Phi(x\otimes x))=S(x\otimes x)=0$, so
$\Phi(x\otimes x)$ is a rank one orthogonal projection. Thus there exists a unit vector $z$
such that $\Phi(x\otimes x)=z\otimes z$.

Let $y\bot x$ be another unit vector of ${\mathcal{H}}$.  For convenience,
denote 
$$
\Phi(x\otimes x)=x'\otimes x' \hbox { and } \Phi(y\otimes y)=y'\otimes y'.
$$ 
Setting $\rho_0=\frac12(x\otimes x+y\otimes y)$, we have
\begin{equation}\label{eq:S=1}
S\bigl(\tfrac12(x'\otimes x'+ y'\otimes y')\bigr)=S\bigl(\Phi(\rho_0)\bigr)=S(\rho_0)=1. 
\end{equation}
In the following, we shall show $ x'\perp y'$.
Let ${{\mathcal{H}}_0}\subseteq{\mathcal{H}}$ denote the two-dimensional space spanned by $x'$ and $y'$.
Then $\Phi(\rho_0)$ can be treated as an operator from ${{\mathcal{H}}_0}$ into ${{\mathcal{H}}_0}$.
Suppose $y'=\alpha x'+\sqrt{1-|\alpha|^2}x'^{\bot}$, where $0\leq|\alpha|\leq 1$ and $x'^{\bot}$ is a
unit vector in ${{\mathcal{H}}_0}$ orthogonal to $x'$.  Then we have
the following matrix forms for the operators $x'\otimes x'$ and $y'\otimes y'$, respectively:
$$ 
x'\otimes x'=\left(\begin{array}{cc}  1 & 0\\  0 &
0 \end{array}\right), \qquad\qquad y'\otimes y'=\left(\begin{array}{cc} |\alpha|^2 & \alpha \sqrt{1-|\alpha|^2}\\  
\overline{\alpha} \sqrt{1-|\alpha|^2}& 1-|\alpha|^2\end{array}\right),
$$
where $\overline{\alpha}$ is the complex conjugate of $\alpha$.
Then  
\begin{equation}\label{eq:mixture}
\frac12( x'\otimes x' +  y'\otimes y')
=\frac12\begin{pmatrix} 
1+|\alpha|^2 & \alpha \sqrt{1-|\alpha|^2}\\  
\overline{\alpha} \sqrt{1-|\alpha|^2} & 1-|\alpha|^2
\end{pmatrix}. 
\end{equation}
The characteristic polynomial of \eqref{eq:mixture} is 
\begin{equation}\label{eq:charpoly}
\lambda^2-\lambda+\frac{1-|\alpha|^2}{4}=0. 
\end{equation} 
Furthermore, Lemma \ref{lem:S-props} (ii) and equation \eqref{eq:S=1} imply that equation \eqref{eq:charpoly} 
has two equal roots 
$\lambda_1=\lambda_2=\frac{1}{2}$, which yields $\alpha=0$. Thus, $ x' \bot y'$ as required. 

From the proof above, we know that the map $\Phi$ sends orthogonal pure states to orthogonal pure states.
Let $\{e_i\}_{i=1}^{\infty}$ be an orthonormal basis of ${\mathcal{H}}$ and $P_n$ the orthogonal projection onto 
the subspace spanned by $\{e_i\}_{i=1}^{n}$. Then 
$$
  \Phi(I)=\lim_{n\to\infty}\Phi(P_n)=\lim_{n\to\infty}\sum_{i=1}^n\Phi(e_i\otimes e_i),
$$ 
so $\Phi(I)$ is an infinite dimensional orthogonal projection.

Let $\rho\in S({\mathcal{H}})$ be injective and $\rho=\sum_{i=1}^\infty \lambda_ix_i\otimes x_i$, where
$\{x_i\}_{i=1}^{\infty}$ is an orthonormal basis and $\lambda_i$ are the eigenvalues of $\rho$.
Then 
$$
\Phi(\rho)=
\lim_{n\to\infty}\Phi\biggl(\sum_{i=1}^n \lambda_ix_i\otimes x_i\biggr)=
\lim_{n\to\infty}\sum_{i=1}^n\lambda_i\Phi(x_i\otimes x_i)=\sum_{i=1}^\infty \lambda_ix_i'\otimes x_i'.
$$
Denote $V_\rho x_i=x_i'$ for all $i$, then $\Phi(\rho)=V_\rho\rho V_\rho^*$.

By Kraus's theorem, we have
$\Phi(\rho)=\sum_{i=1}^{\infty}A_{i}\rho A_{i}^*$, where $A_i\in{\mathcal{ B(H}})$, and $\sum_{i=1}^{\infty}A_{i}^*A_{i}=I$,
as $\Phi$ is trace-preserving.  Then
$$
\sum_{i=1}^{\infty}A_{i}\rho A_{i}^*=V_\rho\rho V_\rho^*,
$$ 
which yields
\begin{equation}\label{eq:V-sum}
\sum_{i=1}^{\infty}V_\rho^*A_{i}\rho A_{i}^*V_\rho=\rho, 
\end{equation} 
for all injective $\rho\in S({\mathcal{H}})$.
Furthermore, $\sum_{i=1}^{\infty} A_{i}A_{i}^*=\Phi(I)\leq I$, so 
\begin{equation*}\sum_{i=1}^{\infty}V_\rho^*A_{i}A_{i}^*V_\rho\leq I \hbox { and }
\sum_{i=1}^{\infty}A_{i}^*V_\rho V_\rho^*A_{i}\leq\sum_{i=1}^{\infty}A_{i}^*A_{i}= I. 
\end{equation*}
Then Lemma \ref{lem:commut} and equation \eqref{eq:V-sum} imply 
$$
\rho\sum_{i=1}^{\infty}V_\rho^*A_{i} A_{i}^*V_\rho=\sum_{i=1}^{\infty}V_\rho^*A_{i}\rho A_{i}^*V_\rho=\rho,
$$ 
and  
\begin{equation}\label{eq:V-sum2}
\rho\sum_{i=1}^{\infty}A_{i}^*V_\rho V_\rho^*A_{i} =\sum_{i=1}^{\infty} A_{i}^*V_\rho\rho V_\rho^*A_{i} =\rho, 
\end{equation} 
so
\begin{equation*} 
I=\sum_{i=1}^{\infty} A_{i}^*V_\rho V_\rho^*A_{i},
\end{equation*}
which yields 
\begin{equation*}
\sum_{i=1}^{\infty} A_{i}^*V_\rho V_\rho^*A_{i}=I=\sum_{i=1}^{\infty}A_{i}^*A_{i}.
\end{equation*}  
Thus
$$
\sum_{i=1}^{\infty}A_{i}^*(I-V_\rho V_\rho^*)A_{i}=0,
$$ 
which  implies $A_{i}^*(I-V_\rho V_\rho^*)=0$, that is 
\begin{equation}\label{eq:AVV}
A_{i}^*=A_{i}^*V_\rho V_\rho^*, \hbox { }\hbox { }\hbox { } \hbox { for all injective }\rho\in S({\mathcal{H}}).
\end{equation}
Using equations \eqref{eq:V-sum} and \eqref{eq:V-sum2}, we get
\begin{equation*}
\sum_{i}A_{i}^* V_\rho (\sum_{j} V_\rho^*A_{j}\rho A_{j}^* V_\rho) V_\rho^*A_i=\rho, 
\end{equation*}  
so equation \eqref{eq:AVV} implies
\begin{equation}\label{eq:AArho}
\sum_{i,j}A_{i}^*A_{j}\rho A_{j}^*A_i=\rho, 
\end{equation}
which is equivalent to $\Phi^{\dagger}\circ \Phi (\rho)=\rho$, for all injective $\rho\in S({\mathcal{H}})$.

For any unit vector $x$, we claim that there exists a sequence of injective states $\rho_n\in S({\mathcal{H}})$ such that
$\rho_n\to x\otimes x$ as $n\to\infty$ (in the weak-$*$ topology). Indeed, let $\{f_i\}_{i=1}^\infty$ be an orthonormal basis of the orthogonal complement subspace of $x$, and
$\rho_n=(1-\frac{1}{n})x\otimes x+\frac{1}{n}\sum_{i=1}^{\infty}2^{-i}f_i\otimes f_i$.
Then $\lim_{n\to\infty}{\rm tr}(|\rho_n-x\otimes x|)=0$, so 
$\lim_{n\to\infty}{\rm tr}[(\rho_n-x\otimes x)X]=0$, for all $X\in {\mathcal{B(H}})$.
Thus $$x\otimes x=\lim_{n\to\infty}\rho_n=\lim_{n\to\infty}\Phi^{\dagger}\circ \Phi (\rho_n)=\Phi^{\dagger}\circ \Phi(x\otimes x),$$
which says that 
$$
\sum_{i,j}A_{i}^*A_{j}(x\otimes x)A_{j}^*A_i=\Phi^{\dagger}\circ \Phi (x\otimes x)=x\otimes x
$$ 
for all rank one projections $x\otimes x$.
Then by Lemma \ref{lem:commut} again, we have $A_{i}^*A_{j}(x\otimes x)=(x\otimes x)A_{i}^*A_{j}$, 
which implies that for all $i,j=1,2\cdots$,
$A_{j}^*A_{i}=\lambda_{ji}I$ and
$A_{i}^*A_{i}=\lambda_{ii}I$, so $\lambda_{ii}>0$ and
$\sum_{i}^{\infty}\lambda_{ii}=1$. It is clear that
$$
 \begin{array}{rcl}
 \left(\begin{array}{c} A_1^*\\A_2^*\\\vdots\\A_n^*\\\vdots\end{array}\right)
 \left(\begin{array}{ccccc} A_1&A_2& \cdots&A_n & \cdots\end{array}\right) &=&\left
  (\begin{array}{ccccc}
   A_{1}^*A_{1}\ \ & \ \ A_{1}^*A_{2}& \ \ \ldots &\ \ A_{1}^*A_{n} & \ \ \ldots \\
    A_{2}^*A_{1}\ \ & \ \ A_{2}^*A_{2}& \ \ \ldots &\ \ A_{2}^*A_{n}& \ \ \ldots \\
    \ldots\ \ &\ \ \ldots&\ \ \ldots &\ \ \ldots  \\
   A_{n}^*A_{1}\ \ & \ \ A_{n}^*A_{2}& \ \ \ldots &\ \ A_{n}^*A_{n}& \ \ \ldots  \\
    \ldots\ \ &\ \ \ldots&\ \ \ldots &\ \ \ldots& \ \ \ldots  \\
  \end{array}\right)\\&=&\left
  (\begin{array}{ccccc}
   \lambda_{11}I\ \ & \ \ \lambda_{12}I& \ \ \ldots &\ \ \lambda_{1n}I & \ \ \ldots \\
    \lambda_{21}I\ \ & \ \ \lambda_{22}I& \ \ \ldots &\ \ \lambda_{2n}I & \ \ \ldots\\
    \ldots\ \ &\ \ \ldots&\ \ \ldots &\ \ \ldots& \ \ \ldots  \\
   \lambda_{n1}I\ \ & \ \ \lambda_{n2}I& \ \ \ldots &\ \ \lambda_{nn}I & \ \ \ldots\\
    \ldots\ \ &\ \ \ldots&\ \ \ldots &\ \ \ldots & \ \ \ldots \\
  \end{array}\right).\end{array}
$$
Denote
\begin{eqnarray*}
M=\left
  (\begin{array}{ccccc}
   \lambda_{11}\ \ & \ \ \lambda_{12}& \ \ \ldots &\ \ \lambda_{1n} & \ \ \ldots \\
    \lambda_{21}\ \ & \ \ \lambda_{22}& \ \ \ldots &\ \ \lambda_{2n}& \ \ \ldots \\
    \ldots\ \ &\ \ \ldots&\ \ \ldots &\ \ \ldots & \ \ \ldots \\
   \lambda_{n1}\ \ & \ \ \lambda_{n2}& \ \ \ldots &\ \ \lambda_{nn} \\ \ldots\ \ &\ \ \ldots&\ \ \ldots &\ \ \ldots & \ \ \ldots \\
  \end{array}\right),
\end{eqnarray*}
then $\lambda_{ji}=\overline{\lambda_{ij}}$ and 
\begin{equation}\label{eq:lambda}
|\lambda_{ij}|^{2}\leq\lambda_{ii}\lambda_{jj}, \text{  for }
1\leq i,j
\end{equation}  
since $M\geq0$.  Further equation \eqref{eq:AArho} implies $\sum_{i,j}A_{j}^*A_{i}
A_{i}^*A_{j}=I$, which yields
$$
\sum_{i,j}|\lambda_{ij}|^{2}=1=\biggl(\sum_{i=1}^{\infty}\lambda_{ii}\biggr)^{2}.
$$  
Then by a direct calculation, we get
$$
\sum_{i\neq j}\mid \lambda_{ij}\mid^{2}=\sum_{i\neq j}\lambda_{ii} \lambda_{jj},
$$
so equation \eqref{eq:lambda} implies 
$$
| \lambda_{ij} |^{2}=\lambda_{ii}\lambda_{jj}, \text{ for } 1\leq i,j\leq  \infty.
$$ 
Thus
$$
  A_{j}^{*}A_{1}=\lambda_{j1}I=\sqrt{\lambda_{11}\lambda_{jj}}\,e^{i\theta_j}I,
$$ 
for $j=1,2\cdots$. We denote $V_j={A_j}/{\sqrt{\lambda_{jj}}}$, so $V_j$ are isometric operators for
$j=1,2\cdots$,  which yields $V_j^*V_1=e^{i\theta_j}I$. Then
$V_j^*V_1V_1^*=e^{i\theta_j}V_1^*$ and $V_jV_j^*V_1=e^{i\theta_j}V_j$. Hence
$$
V_1V_1^*V_j V_j^*V_1V_1^*=V_1V_1^* \text{ and } V_jV_j^*V_1V_1^*V_jV_j^*=V_jV_j^*,
$$
so we get $$V_1V_1^*(I-V_j V_j^*)=0 \hbox{ and } V_jV_j^*(I-V_1V_1^*)=0,$$ then
$V_1V_1^*=V_jV_j^*$, for $j=1,2\cdots$. Thus for all $X\in {\mathcal{B(H}})$,
\begin{align*}
A_{j}X A_{j}^{*}&=V_1V_1^*A_{j}X A_{j}^{*}V_1V_1^*
=V_1\frac1{\lambda_{11}}\,A_1^*A_jXA_j^*A_1\,V_1^*
=\lambda_{jj}V_1X V_1^*,
\end{align*}
which implies
$$
\Phi(X)=A_{1}XA_{1}^{*}+A_{2}X
A_{2}^{*}+\ldots+A_{n}X A_{n}^{ *}+\cdots=V_1X V_1^*.
$$
\end{proof}

The following result is an immediate consequence of
Theorem \ref{thm:isom}. 

\begin{corollary}
Let  $\Phi$ be a bi-stochastic quantum operation. Then
 $S(\Phi(\rho))=S(\rho)$ for all quantum state $\rho\in S(\mathcal{H})$
if and only if there exists a unitary matrix $U$ such that
$\Phi(\rho)=U\rho U^*$.
\end{corollary}

\section*{Acknowledgements}
This work was carried out during Y.L.'s one-year visit to the University of York.
The authors would like to thank an anonymous referees for corrections to the first submitted manuscript version of this paper.
This work is supported by the National Science Foundation of China (Grant No. 10871224, 11001159) and the Fundamental 
Research Funds for the Central Universities (GK201301007), China.


\begin{thebibliography}{11}
 
\bibitem{AU82} P.M.~Alberti and A.~Uhlmann,  
{\em Stochasticity and partial order: doubly stochastic maps and unitary mixing},  
Dordrecht, Boston, 1982.

\bibitem{Ante07}  J. Antezana, P. Massey,  M. Ruiz, and  D. Stojanoff, The Schur-Horn Theorem for operators and frames with prescribed norms and frame operator, {\em Illinois J.~Math.} {\bf 51} (2007) 537-560.

\bibitem{AK06} W.~Arveson and R.V.~Kadison, Diagonals of self-adjoint operators, in: {\em Operator theory, 
operator algebras, and applications}, 
{\em Contemp.~Math.} {\bf 414} (2006) 247-263.

\bibitem{HLP73} G.H.~Hardy, J.E.~Littlewood, and  G.~P$\acute{\rm o}$lya, {\em Inequalities}, 
2nd Ed., Cambridge University Press, 1973.

\bibitem{HMPB11} F.~Hiai, M.~Mosonyi, D.~Petz, C.~B$\acute{e}$ny, Quantum f-divergences and error correction, 
{\em Rev.~Math.~Phys.} {\bf 23} (2011) 691-747.

\bibitem{GM64}  I.~Gohberg and A.~Markus, Some relations between eigenvalues and matrix elements of linear operators,
{\em Mat.~Sb.} {\bf 64} (106) (1964) 481-496 (in Russian); {\em Amer.~Math.~Soc.~Transl.} (2) {\bf 52} (1966) 201-216.

\bibitem{KW10}  V. Kaftal and G. Weiss, An infinite dimensional Schur-Horn theorem and majorization theory, 
{\em J.~Funct.~Anal.} {\bf 259} (2010) 3115-3162.

\bibitem{Kraus71} K.~Kraus, General state changes in quantum theory, {\em Ann.~Phys. (NY)} {\bf 64} (1971) 311-335.

\bibitem {Li11} Y.~Li, Fixed points of dual quantum operations, {\em J.~Math.~Analysis Applic.} {\bf 382} (2011) 172-179.

\bibitem {LiWang12} Y.~Li, Y.~Wang,  Further results on entropy and separability, {\em J.~Phys.~A: Math.~Theor.}
{\bf 45} (2012) 385305.

\bibitem {MS10} L.~Moln\'{a}r, P.~Szokol, Maps on states preserving the relative entropy II, 
{\em Lin.~Alg.~Appl.} {\bf 432} (2010) 3343-3350.

\bibitem{Neumann99} A.~Neumann, An infinite-dimensional generalization of the Schur-Horn convexity theorem,
{\em J.~Funct.~Anal.} {\bf 161} (1999) 418-451.

\bibitem{Nielsen02} M.A.~Nielsen, {\em An introduction of Majorization and its Applications to Quantum Mechanics},  
Lecture Notes, Department of Physics, University of Queensland, Australia (2002).
Available at http://michaelnielsen.org/blog/talks/2002/maj/book.ps.


\bibitem{NC00} M.A.~Nielsen, I.L.~Chuang, {\em Quantum Computation and Quantum Information},
Cambridge University Press, Cambridge (2000).

\bibitem{Petz08} D.~Petz, {\em Quantum Information Theory and Quantum Statistics}, Springer, Berlin, 2008.

\bibitem{Weyl49} H.~Weyl. Inequalities between the two kinds of eigenvalues of a linear transformation,
{\em Proc.~Nat.~Acad.~Sci. (USA)} {\bf 35} (1949) 408-411.

\bibitem{ZW11} L.~Zhang, J.D.~Wu, Von Neumann entropy-preserving quantum operation, 
{\em Phys.~Lett. A} {\bf 375} (2011) 4163-4165.



 \end{thebibliography}
\end{document}